\documentclass[aip,jmp,reprint]{revtex4-2}
\usepackage[colorlinks=true,allcolors={blue}]{hyperref}
\usepackage{amsmath,amsthm,amssymb,enumitem,mathtools,dsfont,graphicx,stmaryrd,ifthen,twoopt,tikz,float}
\usepackage[left=35mm,right=35mm,top=35mm,bottom=35mm]{geometry}
\tikzset{every picture/.style={line width=0.75pt}}
\newcommand{\restrict}[1]{\raise-.5ex\hbox{\ensuremath|}_{#1}}

\newcommand{\N}{\mathbb{N}}
\newcommand{\Z}{\mathbb{Z}}

\newcommand{\R}{\mathbb{R}}

\newcommand{\PPrb}[1]{\mathbb{P}\left[ #1\right]}
\newcommand{\EE}[1]{\mathbb{E}\left[ #1\right]}

\renewcommand{\epsilon}{\varepsilon}

\newcommand{\norm}[1]{\left\| #1\right\|}
\newcommand{\abs}[1]{\left| #1\right|}
\newcommand{\floor}[1]{\left\lfloor #1 \right\rfloor}
\newcommand{\ceil}[1]{\left\lceil #1 \right\rceil}

\newcommand{\hp}[2]{\left\langle #1 ,#2\right\rangle}

\newcommand{\1}{\mathds{1}}
\newcommand{\vect}{\operatorname{Vect}}

\newcommand{\domain}[1]{\ifthenelse{\equal{#1}{}}{\mathcal{D}}{\mathcal{D}\left(#1\right)}}
\theoremstyle{plain}% default
\newtheorem{Theorem}{Theorem}

\newtheorem{Proposition}[Theorem]{Proposition}

\theoremstyle{definition}

\theoremstyle{remark}

\newtheorem*{Remarks}{Remarks}
\begin{document}

\title{The Integrated Density of States of the 1D Discrete Anderson-Bernoulli Model at Rational Energies}
\author{Daniel S\'anchez-Mendoza}
\email{dsanchezmendoza@unistra.fr}
\affiliation{Université de Strasbourg, Institut de Recherche Mathématique Avancée UMR 7501, F-67000 Strasbourg, France.}
\date{\today}
%\keywords{Anderson model, Integrated density of states, Bernoulli distribution.}

\begin{abstract}
We show there is a countable dense set of energies at which the integrated density of states of the 1D discrete Anderson-Bernoulli model can be given explicitly and does not depend on the disorder parameter, provided the latter is above an energy-dependent threshold.
\end{abstract}

\maketitle

\section{Introduction and Results}
Much is known about the Anderson-Bernoulli model on $\Z$. In 1894, Delyon and Souillard\cite{Delyon} gave an elementary proof of the continuity of the integrated density of states (IDS). Spectral localization on the whole spectrum at any disorder was proven in 1987 by Carmona, Klein and Martinelli\cite{Carmona} using Furstenberg's theorem and multi-scale analysis. Later that same year Martinelli and Micheli\cite{Martinelli} gave a lower bound, uniform over the spectrum, on the asymptotic of the Lyapunov exponent as the disorder parameter goes to infinity; and in doing so showed the density of states measure is purely singular continuous if the disorder parameter is large enough. More recently, in 2004, Schulz-Baldes\cite{Schulz} showed the IDS exhibits a strong version of Lifshitz tails in which the Lifshitz constant can be computed at all spectral edges. In this note we add to the previously mentioned articles, and many others, by answering to some extent the questions: what value does the IDS assign to a given energy? and how does its plot look like? More precisely, we show that for every energy $x$ in a countable dense set (which will be called the set of rational energies), the IDS evaluated at $x$ can be given explicitly and it does not depend on the disorder parameter, whenever the latter is above an $x$-dependent critical value.

The operator we are concerned with is
\begin{align*}
H_{p,\zeta}\coloneqq-\Delta+\zeta V_p:\ell^2(\N)&\longrightarrow\ell^2(\N)\\
\phi&\longmapsto (H_{p,\zeta}\phi)(j)\coloneqq\left[2\phi(j)-\phi(j+1)-\phi(j-1)\right]+\zeta V_p(j)\phi(j),
\end{align*}
where the Laplacian has the Dirichlet boundary condition $\phi(0)=0$, the disorder parameter $\zeta$ is assumed to be positive, and the potential $\{V_p(j)\}_{j\in\N}$ is an independent and identically distributed (i.i.d.) sequence of random variables defined over a probability space $(\Omega,\mathcal{F},\mathbb{P})$ following a non-degenerate Bernoulli($p$) distribution, i.e. $\PPrb{V_p(j)=1}=p=1-\PPrb{V_p(j)=0}$.
Defining $H_{p,\zeta}$ on $\N$ instead of $\Z$ simplifies the proof of our main result and makes no difference on the IDS.

The almost sure spectrum of $H_{p,\zeta}$ is 
\begin{equation*}
    \sigma(H_{p,\zeta})=\sigma(-\Delta)+\{0,4\}=[0,4]\cup[\zeta,\zeta+4],
\end{equation*}
and its IDS, denoted $I_{p,\zeta}$, is given by the almost sure limit
\begin{equation*}
   I_{p,\zeta}(x)\coloneqq\lim_{L\to \infty}\frac{1}{L}\#\left\{\lambda\in\sigma\left(H_{p,\zeta}\restrict{\ell^2\left(\{1,\ldots,L\}\right)}\right)\,\middle| \,\lambda\leq x\right\},\qquad x\in\R.
\end{equation*}
We recall that $I_{p,\zeta}(x)$ is non-random, $x\mapsto I_{p,\zeta}(x)$ is a continuous distribution function, and $\zeta\mapsto I_{p,\zeta}(x)$ is decreasing.

Before stating our main result we define the functions
\begin{equation*}
    \beta(x)\coloneqq\frac{\pi}{2\arcsin\left(\sqrt{x}/2\right)},\qquad I_p^{\leq}(x)\coloneqq p^2\sum_{y=1}^{\infty}(1-p)^y\floor{\frac{y+1}{\beta(x)}},\qquad x\in(0,4),
\end{equation*}
where $\floor{\cdot}$ is the floor function. These functions have already appeared in Ref. \onlinecite{DSM}, where $I_p^{\leq}$ had a different series representation; we will see later that they coincide. We also define the set of \textit{rational energies}
\begin{equation*}
    R\coloneqq\left\{\beta^{-1}(b/a)\left(=4\sin^2\left(\frac{\pi a}{2 b}\right)\right)\,\middle|\,a,b\in\N,\,a<b\right\},
\end{equation*}
which is countable and dense in $[0,4]$.
\begin{Theorem}
For all $x\in R$ there is a critical $\zeta_c(x)\in(0,\infty)$ such that
\begin{equation*}
     \zeta\geq\zeta_c(x)\implies I_{p,\zeta}(x)=I_p^\leq(x).
\end{equation*}
If $x=\beta^{-1}(b/a)\in R$ with $a,b\in\N$, $a<b$, $\gcd(a,b)=1$ then $\zeta_c(x)\leq\max\left\{8,\frac{4 b}{\pi}+4\right\}$ and
\begin{equation*}
    I_p^\leq(x)=\frac{p^2}{1-(1-p)^b}\left(\frac{a(1-p)^b}{p}+\sum_{r=0}^{b-1}(1-p)^{r}\floor{\frac{a (r+1)}{b}}\right).
\end{equation*}
Moreover $\displaystyle \lim_{\zeta\to\infty}I_{p,\zeta}(x)=I_p^{\leq}(x)$ for all $x\in(0,4)$.
\end{Theorem}
\begin{Remarks}\leavevmode
\begin{enumerate}
\item We have excluded $x=0$ from the definition of $R$ and $I_p^\leq$ to avoid the singularity of $\beta$, however $I_{p,\zeta}(0)=0$ for all $\zeta\geq0$. We have also excluded $x=4$ since $I_{p,\zeta}(4)=1-p$ for $\zeta\geq4$ but
\begin{equation*}
    p^2\sum_{y=1}^{\infty}(1-p)^y\floor{\frac{y+1}{1}}=(1-p)(1+p)>1-p.
\end{equation*}
\item The unitary map $(U\phi)(j)=(-1)^j\phi(j)$ transforms $H_{p,\zeta}$ as $U H_{p,\zeta} U^*=4+\zeta-\left(-\Delta+\zeta[1-V_p]\right)$. Since $\{1-V_p(j)\}_{j\in\N}$ is an i.i.d. Bernoulli($1-p$) potential, we have 
\begin{equation*}
    I_{p,\zeta}(x)=1-I_{1-p,\zeta}(4+\zeta-x).
\end{equation*}
This allows us to derive an analogous statement for the rational energies of $[\zeta,\zeta+4]\subseteq\sigma(H_{p,\zeta})$.
\item For the cases $a=1$ and $a=b-1$ there is a stronger upper bound on $\zeta_c(\beta^{-1}(b/a))$. From Corollaries 2 and 6 of Ref. \onlinecite{DSM} we have for all $b\in\N\setminus\{1\}$
\begin{equation*}
    \zeta_c(\beta^{-1}(b/1))\leq4\qquad\text{and}\qquad\zeta_c\left(\beta^{-1}\left(b/(b-1)\right)\right)\leq\beta^{-1}\left(b/(b-1)\right).
\end{equation*}
\item The limit $\lim_{\zeta\to\infty}I_{p,\zeta}=I_p^{\leq}$ is only point-wise since $I_{p,\zeta}$ is continuous for every $\zeta$ while $I_p^{\leq}$ is only right-continuous. In particular, there cannot be a finite $\zeta$ for which $I_{p,\zeta}(x)=I_p^\leq(x)$ for all $x\in R$ since this would imply the equality for all $x\in(0,4)$ and therefore uniform convergence.
\end{enumerate}
\end{Remarks}
We can use Theorem 1 and Remark 3 to obtain a granular idea of the plot of $I_{p,\zeta}$ for any given $\zeta\geq8$. Indeed, if $\zeta\geq8$ and we define $n=n(\zeta)\coloneqq\floor{\frac{\pi(\zeta-4)}{4}}\in\N$ and the set of energies
\begin{align*}
    R_n\coloneqq &\left\{\beta^{-1}(b/a)\,\middle|\,a,b\in\N,\,\,a<b\leq n\right\}\cup\left\{\beta^{-1}(b/1)\,\middle|\,b\in\N\setminus\{1\}\right\}\cup\left\{\beta^{-1}\left(b/(b-1)\right)\,\middle|\,b\in\N\setminus\{1\}\right\}\subseteq R,
\end{align*}
we have $I_{p,\zeta}(x)=I_p^\leq(x)$ for all $x\in R_n$, as shown in Figure \ref{Fig1}. The first set in the definition of $R_n$ is finite while the other two are countable and accumulate towards $0$ and $4$ respectively. Naturally, as $\zeta$ increases so does $n$ and $R_n\uparrow R$.
\begin{figure}[H]
    \centering
    \includegraphics[width=14cm]{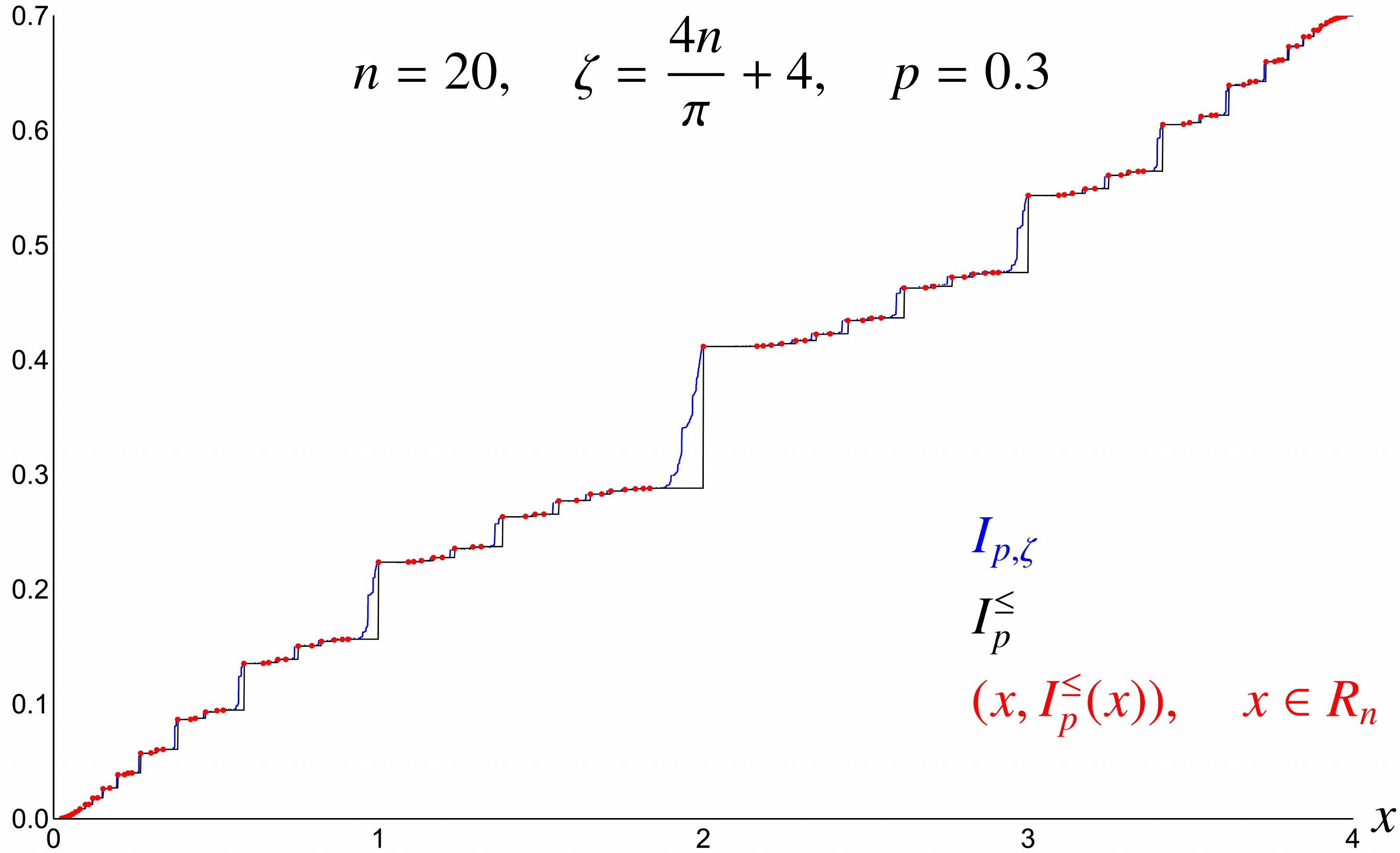}
    \caption{Plot of $I_{p,\zeta}$, $I_{p}^\leq$ and the points $\{(x,I_p^\leq(x))\,|\,x\in R_n\}$ for $n=20$, $\zeta=\frac{4 n}{\pi}+4$, $p=0.3$. $I_{p,\zeta}$ was computed numerically from a $10^5\times10^5$ matrix.}
    \label{Fig1}
\end{figure}
The proof of Theorem 1 consists of bounding $I_{p,\zeta}$, from above and below, by the IDS of a direct sum of i.i.d. random operators whose spectra is explicit or we can approximate very well. This is done by applying a modified Dirichlet-Neumann bracketing to the finite volume restriction of $H_{p,\zeta}$, just as in Ref. \onlinecite{DSM}. We deviate form the aforementioned article on the estimates of the eigenvalues relevant to the upper bound of $I_{p,\zeta}$. There, we used a $\zeta$-independent estimate which led to Remark 3 (see above), while here we use a $\zeta$-dependent one, namely Proposition \ref{PEV}. It is worth noting that a stronger upper bound on $\zeta_c(x)$ than the one given in Theorem 1 may be achieved by refining Proposition \ref{PEV}. In particular, a more careful treatment of equation \eqref{EVE2} and its solutions may lead to an upper bound on $\zeta_c(\beta^{-1}(b/a))$ that depends on $a$ and not just $b$. 

\section{Proof of Theorem 1}
We start by giving all the necessary definitions and notations.

We define two sequences of random variables
\begin{align*}
    L_1&\coloneqq\min\{j>0\,|\,V_p(j)=1\},& Y_1&\coloneqq L_1-1,\\
    L_{n+1}&\coloneqq \min\{j>L_n\,|\,V_p(j)=1\},& Y_{n+1}&\coloneqq L_{n+1}-L_n-1,
\end{align*}
which give respectively, the position of the $1$'s of $V_p$ and the number of $0$'s between them, as shown in Figure \ref{Fig2}. The $Y_i$ are i.i.d. following a geometric distribution $\PPrb{Y_i=y}=(1-p)^y p$ for $y\in \N\cup\{0\}$, and by definition $L_n=n+\sum_{i=1}^n Y_i$. By applying the Law of Large Numbers we obtain $\lim_{n\to\infty}\frac{L_n}{n}=1+\EE{Y_1}=\frac{1}{p}$, and therefore we can use the random subsequence $\{L_n\}_{n\in\N}$ in the definition of $I_{p,\zeta}(x)$.
\begin{figure}[H]
    \centering
    \includegraphics{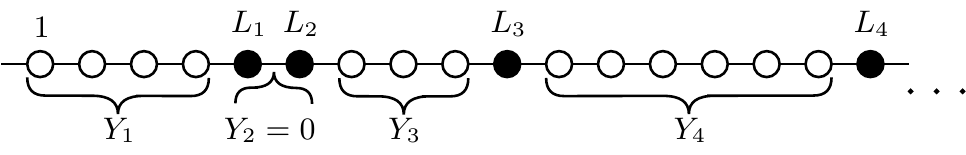}
    \caption{A possible realization of $H_{p,\zeta}$. The Laplacian is given by the graph structure and the potential by the color of the vertices. White (resp. black) vertices represent points where $V_p(j)=0$ (resp. $V_p(j)=1$). Reproduced from Ref. \onlinecite{DSM}, with the permission of AIP Publishing.}
    \label{Fig2}
\end{figure}

We order the eigenvalues of any self-adjoint $n$-dimensional operator $O$ increasingly allowing for multiplicities
\begin{equation*}
    \lambda_1(O)\leq\lambda_2(O)\leq\cdots\leq\lambda_n(O),
\end{equation*}
and introduce the $n\times n$ matrices $A_{n}(i,j)\coloneqq\delta_{1,i}\delta_{1,j}+\delta_{n,i}\delta_{n,j}$ and 
\begin{equation*}
-\Delta_{n}\coloneqq\begin{pmatrix}
2 & -1 & & \\
-1 & \ddots & \ddots & \\
& \ddots & \ddots & -1 \\
 & & -1 & 2 \end{pmatrix},\qquad\sigma(-\Delta_{n})=\left\{\lambda_k(-\Delta_{n})= 4 \sin^2\left(\frac{\pi k}{2(n+1)}\right)\,\middle|\,k=1,\ldots,n\right\}.
\end{equation*}

We identify $H_{p,\zeta}\restrict{\ell^2\left(\{1,\ldots,L_n\}\right)}$ with
$-\Delta_{L_n}+\zeta V_p$ (where the restriction of $V_p$ is implicit) and remark that the continuity of $x\mapsto I_{p,\zeta}(x)$ means it can be computed by counting eigenvalues less or equal ($\leq$) or less ($<$) than $x$:
\begin{equation*}
    I_{p,\zeta}(x)=\lim_{n\to \infty}\frac{1}{L_n}\#\left\{\lambda\in\sigma\left(-\Delta_{L_n}+\zeta V_p\right)\,\middle| \,\lambda\leq x\right\}=\lim_{n\to \infty}\frac{1}{L_n}\#\left\{\lambda\in\sigma\left(-\Delta_{L_n}+\zeta V_p\right)\,\middle| \,\lambda< x\right\}.
\end{equation*}

The lower bound of $I_{p,\zeta}$ just requires an application of the Cauchy Eigenvalue Interlacing Theorem to $-\Delta_{L_n}+\zeta V_p$. Indeed, if we delete from $-\Delta_{L_n}+\zeta V_p$ the $j$-th row and $j$-th column for all $j\in\{1,\ldots,L_n\}$ such that $V_p(j)=1$, the resulting sub-matrix is $\bigoplus_{i=1}^n-\Delta_{Y_i}$ and therefore
\begin{equation*}
    \lambda_k(-\Delta_{L_n}+\zeta V_p)\leq\lambda_k\left(\bigoplus_{i=1}^n-\Delta_{Y_i}\right),\qquad k=1,\ldots,\sum_{i=1}^n Y_i.
\end{equation*}
By counting eigenvalues less or equal ($\leq$) than $x$ and applying the Law of Large Numbers we obtain the lower bound
\begin{align}\label{LB1}
I_{p,\zeta}(x)&\geq\lim_{n\to \infty}\frac{1}{L_n}\#\left\{\lambda\in\sigma\left(\bigoplus_{i=1}^n-\Delta_{Y_i}\right)\,\middle| \,\lambda\leq x\right\}\nonumber\\
&=p\,\EE{\#\left\{\lambda\in\sigma\left(-\Delta_{Y_1}\right)\,\middle| \,\lambda\leq x\right\}},\qquad x\in\R,\,\zeta\geq0.
\end{align}
The limit on \eqref{LB1} is the definition given to $I_p^\leq$ in Ref. \onlinecite{DSM}, we can easily check that they coincide for $x\in(0,4)$:
\begin{align*}
    \lim_{n\to \infty}\frac{1}{L_n}\#\left\{\lambda\in\sigma\left(\bigoplus_{i=1}^n-\Delta_{Y_i}\right)\,\middle| \,\lambda\leq x\right\}&=p\,\EE{\#\left\{\lambda\in\sigma\left(-\Delta_{Y_1}\right)\,\middle| \,\lambda\leq x\right\}}\\
    &=p\sum_{y=0}^\infty\PPrb{Y_1=y} \#\left\{\lambda\in\sigma\left(-\Delta_{y}\right)\,\middle| \,\lambda\leq x\right\}\\
    &=p^2\sum_{y=1}^\infty(1-p)^y \max\{k\in\N\,|\,k\leq y,\, \lambda_k(-\Delta_y)\leq x\}\\
    &=p^2\sum_{y=1}^\infty(1-p)^y \max\left\{k\in\N\,\middle|\,k\leq\min\left\{y,\frac{y+1}{\beta(x)}\right\}\right\}\\
    &=p^2\sum_{y=1}^\infty(1-p)^y\floor{\frac{y+1}{\beta(x)}}.
\end{align*}
Hence we have shown
\begin{equation}\label{LB2}
   I_p^\leq(x)\leq I_{p,\zeta}(x),\qquad x\in(0,4),\,\zeta\geq0.
\end{equation}

The upper bound is a bit more involved. From $-\Delta_{L_n}+\zeta V_p$ we define a new (dimensionally larger) operator $-\Delta_{L_n+n}+\frac{\zeta}{2}V'$ where $V'$ is constructed by doubling each point at which $V_p(j)=1$ while maintaining the $Y_i$'s, as shown in Figure \ref{Fig3}. To be precise,
\begin{equation*}
    V'(j)\coloneqq\sum_{k=1}^\infty\left(\delta_{L_k+k-1,j}+\delta_{L_k+k,j}\right)\quad\text{whereas}\quad V_p(j)=\sum_{k=1}^\infty \delta_{L_k,j}.
\end{equation*}
In order to compare these two operators we define the linear map
$T:\ell^2\left(\{1,\ldots,L_n\}\right)\longrightarrow\ell^2\left(\{1,\ldots,L_n+n\}\right)$,
\begin{equation*}
    (T\phi)(j)\coloneqq\begin{cases}\phi(j-k),&\text{if }\,L_k+k+1\leq j \leq L_{k+1}+(k+1)-1,\\
    \phi(L_k),&\text{if }\,j=L_k+k,
    \end{cases}
\end{equation*}
with the convention $L_0=0$, which assigns to $T\phi$ the same values of $\phi$ according to Figure \ref{Fig3}. For all $\phi\in\ell^2\left(\{1,\ldots,L_n\}\right)$ the map $T$ satisfies
\begin{align*}
    \hp{T\phi}{\left(-\Delta_{L_n+n}+\frac{\zeta}{2}V'\right)T\phi}&=\hp{\phi}{(-\Delta_{L_n}+\zeta V_p)\phi},\\
    \norm{T\phi}&\geq\norm{\phi}.\qquad (T\text{ is injective})
\end{align*}
Let $\phi_i$ be the normalized eigenvector associated to $\lambda_i\big(-\Delta_{L_n}+\zeta V_p\big)$. Then, by the Min-Max Principle we have for $k\leq L_n$
\begin{align*}
    \lambda_k\left(-\Delta_{L_n+n}+\frac{\zeta}{2}V'\right)&\leq\sup_{\phi\in\vect\{\phi_i,\ldots\phi_k\}\setminus\{0\}}\frac{\hp{T\phi}{(-\Delta_{L_n+n}+\frac{\zeta}{2}V')T\phi}}{\norm{T\phi}^2}\\
    &\leq\sup_{\phi\in\vect\{\phi_i,\ldots\phi_k\}\setminus\{0\}}\frac{\hp{\phi}{(-\Delta_{L_n}+\zeta V_p)\phi}}{\norm{\phi}^2}=\lambda_k\left(-\Delta_{L_n}+\zeta V_p\right).
\end{align*}
\begin{figure}[H]
    \centering
    \includegraphics{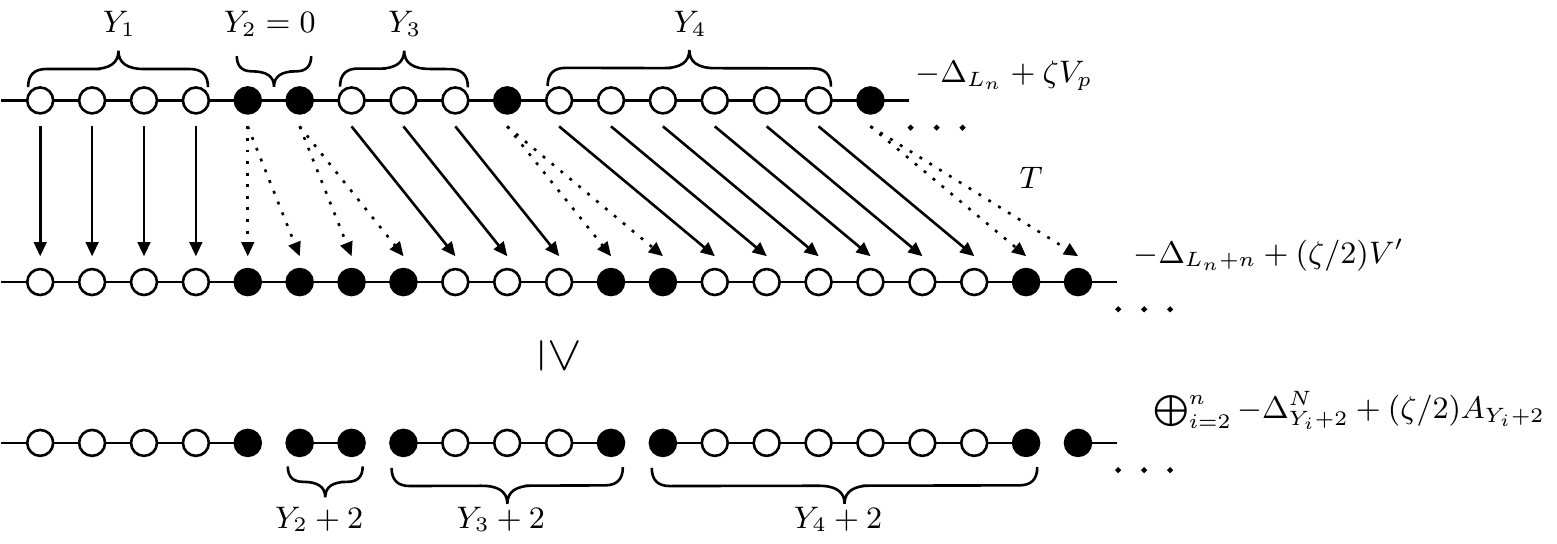}
    \caption{The first two rows show the of construction of $V'$ and the action of $T$. From the second to the third row we have deleted edges, which lowers the operator and decomposes it into a direct sum. Partially reproduced from Ref. \onlinecite{DSM}, with the permission of AIP Publishing}
    \label{Fig3}
\end{figure}

We can construct form $-\Delta_{L_n+n}+\frac{\zeta}{2}V'$ an operator with even lower eigenvalues by disconnecting each $Y_i$ (except $Y_1$) together with its two adjacent points at the cost of having Neumann boundary conditions on the Laplacians, as shown in Figure \ref{Fig3}. Since $Y_1$ has no point to its left and the right-most point ($j=L_n+n$) ends up isolated, we have a boundary term of dimension $Y_1+2$. This, together with the previous lower bound on $\lambda_k(-\Delta_{L_n}+\zeta V_p)$ and the fact that we can write the Neumann Laplacian as $-\Delta^N_n=-\Delta_n-A_n$, gives
\begin{equation*}
    \lambda_k\left((\text{Boundary term})\oplus\bigoplus_{i=2}^n-\Delta_{Y_i+2}+\left(\frac{\zeta}{2}-1\right) A_{Y_i+2}\right)\leq\lambda_k(-\Delta_{L_n}+\zeta V_p),\qquad k=1,\ldots,L_n.
\end{equation*}
Counting eigenvalues less ($<$) than $x$ we obtain
\begin{align}\label{UB1}
    I_{p,\zeta}(x)&\leq \lim_{n\to \infty}\frac{1}{L_n}\#\left\{\lambda\in\sigma\left(\bigoplus_{i=2}^n-\Delta_{Y_i+2}+\left(\frac{\zeta}{2}-1\right) A_{Y_i+2}\right)\,\middle| \,\lambda < x\right\}+\lim_{n\to \infty}\frac{Y_1+2}{L_n}\nonumber\\
    &=p\,\EE{\#\left\{\lambda\in\sigma\left(-\Delta_{Y_1+2}+\left(\frac{\zeta}{2}-1\right) A_{Y_1+2}\right)\,\middle| \,\lambda< x\right\}},\qquad x\in\R,\,\zeta\geq0.
\end{align}

To further bound \eqref{UB1} we need to estimate the eigenvalues that appear on it; which is the purpose of the next proposition. These eigenvalues are always simple since their eigenvectors satisfy a second order difference equation with two boundary conditions.
\begin{Proposition}\label{PEV}
Let $n\in\N\cup\{0\}$ and define  $\mu_{k,n+2}(t)\coloneqq\lambda_k(-\Delta_{n+2}+t A_{n+2})$. If $t\geq3$ then:
\begin{enumerate}[label=\normalfont\roman*)]
    \item $0<\lambda_{k}(-\Delta_{n+2})\leq\mu_{k,n+2}(t)\leq\lambda_{k}(-\Delta_{n})<4$ for $1\leq k\leq n$.
    \item $4\leq\mu_{n+1,n+2}(t)< \mu_{n+2,n+2}(t)$.
    \item $\displaystyle \mu_{k,n+2}(t)\geq4\sin^2\left(\frac{\pi}{2(n+1)}\left[k-\frac{2}{\pi(t-1)}\right]\right)$ for $1\leq k\leq n$.
\end{enumerate}
\end{Proposition}
\begin{proof}\leavevmode
\begin{enumerate}[label=\roman*)]
    \item The lower bound on $\mu_{k,n+2}(t)$ follows from $t A_{n+2}\geq0$, while the upper one follows from the Cauchy Eigenvalue Interlacing Theorem by deleting from $-\Delta_{n+2}+t A_{n+2}$ the rows (and columns) where $t$ appears.
    \item For $n=0$ we compute directly 
    \begin{equation*}
        \sigma(-\Delta_{2}+t A_{2})=\sigma\begin{pmatrix}2+t&-1\\-1&2+t\end{pmatrix}=\{1+t,3+t\}.
    \end{equation*}
    For $n\geq1$ we apply the Min-Max Principle (Max-Min in this case):
    \begin{align*}
        \mu_{n+1,n+2}(t)&\geq\min_{\substack{\phi\in\vect\{e_1,e_{n+2}\}\\\norm{\phi}=1}}\hp{\phi}{(-\Delta_{n+2}+t A_{n+2})\phi}\\
        &=\min_{\substack{\phi\in\vect\{e_1,e_{n+2}\}\\\norm{\phi}=1}}(2+t)\norm{\phi}^2=2+t\geq4,
    \end{align*}
    where $e_i$ denotes the canonical basis of $\ell^2\left(\{1,\ldots,n+2\}\right)$.
    \item We recall that the characteristic polynomial of $-\Delta_n$ can be written as
    \begin{equation*}
        \det(-\Delta_n-x)=(-1)^n U_n\left(\frac{x-2}{2}\right),
    \end{equation*}
    where $U_n$ is the $n$-th Chebyshev polynomial of the second kind. For completeness we list here the properties of $U_n$ (see Section 1.2.2 of Ref. \onlinecite{Chebyshev}) that we will need:
    \begin{itemize}
        \item Recurrent definition:
        \begin{equation*}
            U_0(x)\coloneqq1,\quad  U_1(x)\coloneqq2x,\quad U_{n+1}(x)\coloneqq2x U_n(x)-U_{n-1}(x).
        \end{equation*}
        \item Parity:
        \begin{equation*}
            U_n(-x)=(-1)^n U_n(x).
        \end{equation*}
        \item Image of $(-1,1)$:
        \begin{equation*}
            U_n(\cos\theta)=\frac{\sin((n+1)\theta)}{\sin\theta}.
        \end{equation*}
    \end{itemize}
    Now we start with the proof. A straight forward computation shows that we can expand the characteristic polynomial of $-\Delta_{n+2}+t A_{n+2}$ as
    \begin{multline*}
    \det(-\Delta_{n+2}+t A_{n+2}-x)\\
    \begin{aligned}
        &=(2+t-x)^2\det(-\Delta_{n}-x)-2(2+t-x)\det(-\Delta_{n-1}-x)+\det(-\Delta_{n-2}-x)\\&=(-1)^n\left[(2+t-x)^2U_n\left(\frac{x-2}{2}\right)+2(2+t-x)U_{n-1}\left(\frac{x-2}{2}\right)+U_{n-2}\left(\frac{x-2}{2}\right)\right].
    \end{aligned}
    \end{multline*}
    By introducing the change of variable $x'\coloneqq \frac{x-2}{2}$ and using twice the recurrent definition of $U_n$, the previous expression can be reduced to
    \begin{equation*}
         \det(-\Delta_{n+2}+t A_{n+2}-x)=(-1)^n\left[(t^2-1)U_n(x')-2(t-x')U_{n+1}(x')\right].
    \end{equation*}
    By i) and ii), $-\Delta_{n+2}+t A_{n+2}$ has exactly $n$ simple eigenvalues in $(0,4)$. With this in mind, we introduce a parameter $\theta\in(0,\pi)$ and notice, by evaluating the characteristic polynomial at $x'=-\cos(\theta)$, that $ 4\sin^2(\theta/2)\in\sigma(-\Delta_{n+2}+t A_{n+2})$ if and only if
    \begin{equation}\label{EVE1}
        (t^2-1)\sin((n+1)\theta)=-2(t+\cos\theta)\sin((n+2)\theta).
    \end{equation}
    The condition $t\geq3$ and the trigonometric identity
    \begin{equation*}
        \sin((n+2)\theta)=\sin((n+1)\theta)\cos\theta+\cos((n+1)\theta)\sin\theta
    \end{equation*}
    guaranty that there is no solution to \eqref{EVE1} in the set $\frac{\pi}{n+1}\Z$, hence we can rewrite the equation as
    \begin{align}
        \frac{\sin((n+2)\theta)}{\sin((n+1)\theta)}&=-\frac{t^2-1}{2(t+\cos\theta)},\nonumber\\
      \cos\theta+\cot((n+1)\theta)\sin\theta&=-\frac{t^2-1}{2(t+\cos\theta)},\nonumber\\
      \label{EVE2}\tan((n+1)\theta)&=-\frac{2(t+\cos\theta)\sin\theta}{t^2+2t\cos\theta+\cos(2\theta)}.
    \end{align}
    To abbreviate we define
    \begin{equation*}
        f_t(\theta)\coloneqq\frac{2(t+\cos\theta)\sin\theta}{t^2+2t\cos\theta+\cos(2\theta)},\qquad\theta\in(0,\pi),
    \end{equation*}
    and remark that $t\geq3$ implies $0<f_t(\theta)<\infty$.
    
    Applying arc-tangent to \eqref{EVE2} and considering that i) actually constrains $\theta$ to be in $\left[\frac{\pi }{n+3},\frac{\pi n}{n+1}\right]$, we conclude for $k=1,\ldots,n$ that
    \begin{equation*}
        \mu_{k,n+2}(t)=4\sin^2(\theta_k/2)\quad\text{where $\theta_k$ is defined by}\quad\theta_k=\frac{\pi k-\arctan f_t(\theta_k)}{n+1}.
    \end{equation*}
    The existence of $\theta_k$ is a consequence of tangent going from $-\infty$ to $+\infty$ over a period. Uniqueness comes from $\abs{\theta_{k+1}-\theta_k}\geq\frac{\pi}{2(n+1)}$ and the fact that $-\Delta_{n+2}+t A_{n+2}$ has exactly $n$ eigenvalues in $(0,4)$.
    
    After bounding uniformly $f_t(\theta)$
    \begin{align*}
        \sup_{\theta\in(0,\pi)}f_t(\theta)&= \sup_{\theta\in(0,\pi)}\frac{2(t+\cos\theta)\sin\theta}{t^2+2t\cos\theta+\cos(2\theta)}\\
        &\leq2\sup_{\theta\in(0,\pi)}\frac{t+\cos\theta}{t^2+2t\cos\theta+\cos(2\theta)}\\
        &=2\frac{t+\cos\theta}{t^2+2t\cos\theta+\cos(2\theta)}\Big|_{\theta=\pi}=\frac{2}{t-1},
    \end{align*}
    and using the inequality $\arctan(x)\leq x$ for $x\geq0$ we obtain
    \begin{equation*}
         \mu_{k,n+2}(t)\geq4\sin^2\left(\frac{\pi}{2(n+1)}\left[k-\frac{2}{\pi(t-1)}\right]\right),\qquad k=1,\ldots,n.\qedhere
    \end{equation*}
\end{enumerate}
\end{proof}

We now use ii) and iii) of Proposition \ref{PEV} with $t=\frac{\zeta}{2}-1\geq3$ (equivalently $\zeta\geq8$) to further bound \eqref{UB1} for $x\in(0,4)$:
\begin{align}\label{UB2}
    I_{p,\zeta}(x)&\leq p\,\EE{\#\left\{\lambda\in\sigma\left(-\Delta_{Y_1+2}+t A_{Y_1+2}\right)\,\middle| \,\lambda< x\right\}}\nonumber\\
    &=p^2\sum_{y=0}^\infty(1-p)^y \max\{k\in\N\,|\,k\leq y+2,\, \mu_{k,y+2}(t)< x\}\nonumber\\
    &=p^2\sum_{y=1}^\infty(1-p)^y \max\{k\in\N\,|\,k\leq y,\, \mu_{k,y+2}(t)< x\}\nonumber\\
    &\leq p^2\sum_{y=1}^\infty(1-p)^y \max\left\{k\in\N\,\middle|\,k\leq y,\, 4\sin^2\left(\frac{\pi}{2(y+1)}\left[k-\frac{4}{\pi(\zeta-4)}\right]\right)< x\right\}\nonumber\\
    &\leq p^2\sum_{y=1}^\infty(1-p)^y \max\left\{k\in\N\,\middle|\,k<\frac{y+1}{\beta(x)}+\frac{4}{\pi(\zeta-4)}\right\}\nonumber\\
    &=p^2\sum_{y=1}^{\infty}(1-p)^y\left(\ceil{\frac{y+1}{\beta(x)}+\frac{4}{\pi(\zeta-4)}}-1\right),\qquad x\in(0,4),\,\zeta\geq8,
\end{align}
where $\ceil{\cdot}$ is the ceiling function.
From \eqref{LB2}, \eqref{UB2}, the inequality $\ceil{\cdot}-1\leq\floor{\cdot}$, the right continuity of $\floor{\cdot}$, and the Dominated Convergence Theorem we conclude
\begin{equation*}
    \lim_{\zeta\to\infty}I_{p,\zeta}(x)=I_p^{\leq}(x),\quad x\in(0,4).
\end{equation*}

Now fix $x=\beta^{-1}(b/a)\in R$ with $a,b\in\N$, $a<b$, $\gcd(a,b)=1$ and further assume $\zeta\geq\frac{4b}{\pi}+4$. This implies
\begin{align*}
    \ceil{\frac{a(y+1)}{b}+\frac{4}{\pi(\zeta-4)}}-1&=\floor{\frac{a(y+1)}{b}}+\ceil{\left\{\frac{a(y+1)}{b}\right\}+\frac{4}{\pi(\zeta-4)}}-1\\
    &\leq \floor{\frac{a(y+1)}{b}}+\ceil{\frac{b-1}{b}+\frac{1}{b}}-1=\floor{\frac{a(y+1)}{b}},\qquad y\in\N,
\end{align*}
where we have used the fractional part $\{x\}=x-\floor{x}$. The last inequality, together with \eqref{LB2} and \eqref{UB2}, gives
\begin{equation*}
    I_{p,\zeta}(x)=I_p^\leq(x)=p^2\sum_{y=1}^{\infty}(1-p)^y\floor{\frac{a(y+1)}{b}}.
\end{equation*}
The existence of $\zeta_c(x)$ and the bound $\zeta_c(x)\leq\max\left\{8,\frac{4 b}{\pi}+4\right\}$ follow.

It only remains to prove that we can replace the infinite series by a finite sum. This is simply done by using the euclidean division $y=b n+r$ and splitting the series over all possible remainders:
\begin{align*}
    I_p^\leq(x)&=p^2\sum_{y=1}^{\infty}(1-p)^y\floor{\frac{a(y+1)}{b}}\\
    &=p^2\sum_{r=0}^{b-1}(1-p)^{r}\sum_{n=0}^\infty(1-p)^{b n}\left(a n+\floor{\frac{a (r+1)}{b}}\right)\\
    &=p^2\sum_{r=0}^{b-1}(1-p)^{r}\left(\frac{a(1-p)^b}{[1-(1-p)^b]^2}+\floor{\frac{a (r+1)}{b}}\frac{1}{1-(1-p)^b}\right)\\
    &=\frac{p^2}{1-(1-p)^b}\left(\frac{a(1-p)^b}{p}+\sum_{r=0}^{b-1}(1-p)^{r}\floor{\frac{a (r+1)}{b}}\right).
\end{align*}

\begin{acknowledgments}
This work has benefitted from support provided by the University of Strasbourg Institute for Advanced Study (USIAS), within the French national programme “Investment for the future” (IdEx-Unistra).
\end{acknowledgments}

\section*{Data Availability}
Data sharing is not applicable to this article as no new data were created or analyzed in this study.

\section*{References}
%aipnum4-2.bst 2019-01-14 (MD) hand-edited version of apsrev4-1.bst
%Control: key (0)
%Control: author (8) initials jnrlst
%Control: editor formatted (1) identically to author
%Control: production of article title (0) allowed
%Control: page (1) range
%Control: year (1) truncated
%Control: production of eprint (0) enabled
%

\begin{thebibliography}{6}%
\makeatletter
\providecommand \@ifxundefined [1]{%
 \@ifx{#1\undefined}
}%
\providecommand \@ifnum [1]{%
 \ifnum #1\expandafter \@firstoftwo
 \else \expandafter \@secondoftwo
 \fi
}%
\providecommand \@ifx [1]{%
 \ifx #1\expandafter \@firstoftwo
 \else \expandafter \@secondoftwo
 \fi
}%
\providecommand \natexlab [1]{#1}%
\providecommand \enquote  [1]{``#1''}%
\providecommand \bibnamefont  [1]{#1}%
\providecommand \bibfnamefont [1]{#1}%
\providecommand \citenamefont [1]{#1}%
\providecommand \href@noop [0]{\@secondoftwo}%
\providecommand \href [0]{\begingroup \@sanitize@url \@href}%
\providecommand \@href[1]{\@@startlink{#1}\@@href}%
\providecommand \@@href[1]{\endgroup#1\@@endlink}%
\providecommand \@sanitize@url [0]{\catcode `\\12\catcode `\$12\catcode
  `\&12\catcode `\#12\catcode `\^12\catcode `\_12\catcode `\%12\relax}%
\providecommand \@@startlink[1]{}%
\providecommand \@@endlink[0]{}%
\providecommand \url  [0]{\begingroup\@sanitize@url \@url }%
\providecommand \@url [1]{\endgroup\@href {#1}{\urlprefix }}%
\providecommand \urlprefix  [0]{URL }%
\providecommand \Eprint [0]{\href }%
\providecommand \doibase [0]{https://doi.org/}%
\providecommand \selectlanguage [0]{\@gobble}%
\providecommand \bibinfo  [0]{\@secondoftwo}%
\providecommand \bibfield  [0]{\@secondoftwo}%
\providecommand \translation [1]{[#1]}%
\providecommand \BibitemOpen [0]{}%
\providecommand \bibitemStop [0]{}%
\providecommand \bibitemNoStop [0]{.\EOS\space}%
\providecommand \EOS [0]{\spacefactor3000\relax}%
\providecommand \BibitemShut  [1]{\csname bibitem#1\endcsname}%
\let\auto@bib@innerbib\@empty
%</preamble>
\bibitem [{\citenamefont {Delyon}\ and\ \citenamefont
  {Souillard}(1984)}]{Delyon}%
  \BibitemOpen
  \bibfield  {author} {\bibinfo {author} {\bibfnamefont {F.}~\bibnamefont
  {Delyon}}\ and\ \bibinfo {author} {\bibfnamefont {B.}~\bibnamefont
  {Souillard}},\ }\bibfield  {title} {\enquote {\bibinfo {title} {Remark on the
  continuity of the density of states of ergodic finite difference
  operators},}\ }\href {https://doi.org/10.1007/BF01209306} {\bibfield
  {journal} {\bibinfo  {journal} {Communications in Mathematical Physics}\
  }\textbf {\bibinfo {volume} {94}},\ \bibinfo {pages} {289--291} (\bibinfo
  {year} {1984})}\BibitemShut {NoStop}%
\bibitem [{\citenamefont {Carmona}, \citenamefont {Klein},\ and\ \citenamefont
  {Martinelli}(1987)}]{Carmona}%
  \BibitemOpen
  \bibfield  {author} {\bibinfo {author} {\bibfnamefont {R.}~\bibnamefont
  {Carmona}}, \bibinfo {author} {\bibfnamefont {A.}~\bibnamefont {Klein}},\
  and\ \bibinfo {author} {\bibfnamefont {F.}~\bibnamefont {Martinelli}},\
  }\bibfield  {title} {\enquote {\bibinfo {title} {Anderson localization for
  bernoulli and other singular potentials},}\ }\href
  {https://doi.org/10.1007/BF01210702} {\bibfield  {journal} {\bibinfo
  {journal} {Communications in Mathematical Physics}\ }\textbf {\bibinfo
  {volume} {108}},\ \bibinfo {pages} {41--66} (\bibinfo {year}
  {1987})}\BibitemShut {NoStop}%
\bibitem [{\citenamefont {Martinelli}\ and\ \citenamefont
  {Micheli}(1987)}]{Martinelli}%
  \BibitemOpen
  \bibfield  {author} {\bibinfo {author} {\bibfnamefont {F.}~\bibnamefont
  {Martinelli}}\ and\ \bibinfo {author} {\bibfnamefont {L.}~\bibnamefont
  {Micheli}},\ }\bibfield  {title} {\enquote {\bibinfo {title} {On the
  large-coupling-constant behavior of the liapunov exponent in a binary
  alloy},}\ }\href {https://doi.org/10.1007/BF01010397} {\bibfield  {journal}
  {\bibinfo  {journal} {Journal of statistical physics}\ }\textbf {\bibinfo
  {volume} {48}},\ \bibinfo {pages} {1--18} (\bibinfo {year}
  {1987})}\BibitemShut {NoStop}%
\bibitem [{\citenamefont {Schulz-Baldes}(2003)}]{Schulz}%
  \BibitemOpen
  \bibfield  {author} {\bibinfo {author} {\bibfnamefont {H.}~\bibnamefont
  {Schulz-Baldes}},\ }\bibfield  {title} {\enquote {\bibinfo {title} {Lifshitz
  tails for the 1{D} {B}ernoulli-{A}nderson model},}\ }\href
  {http://math-mprf.org/journal/articles/id1005/} {\bibfield  {journal}
  {\bibinfo  {journal} {Markov Processes and Related Fields}\ }\textbf
  {\bibinfo {volume} {9}},\ \bibinfo {pages} {795--802} (\bibinfo {year}
  {2003})}\BibitemShut {NoStop}%
\bibitem [{\citenamefont {Sánchez-Mendoza}(2021)}]{DSM}%
  \BibitemOpen
  \bibfield  {author} {\bibinfo {author} {\bibfnamefont {D.}~\bibnamefont
  {Sánchez-Mendoza}},\ }\bibfield  {title} {\enquote {\bibinfo {title} {Sharp
  bounds for the integrated density of states of a strongly disordered 1d
  anderson–bernoulli model},}\ }\href {https://doi.org/10.1063/5.0037707}
  {\bibfield  {journal} {\bibinfo  {journal} {Journal of Mathematical Physics}\
  }\textbf {\bibinfo {volume} {62}},\ \bibinfo {pages} {072107} (\bibinfo
  {year} {2021})}\BibitemShut {NoStop}%
\bibitem [{\citenamefont {Mason}\ and\ \citenamefont
  {Handscomb}(2003)}]{Chebyshev}%
  \BibitemOpen
  \bibfield  {author} {\bibinfo {author} {\bibfnamefont {J.~C.}\ \bibnamefont
  {Mason}}\ and\ \bibinfo {author} {\bibfnamefont {D.~C.}\ \bibnamefont
  {Handscomb}},\ }\href@noop {} {\emph {\bibinfo {title} {Chebyshev
  polynomials}}}\ (\bibinfo  {publisher} {Chapman \& Hall/CRC, Boca Raton,
  FL},\ \bibinfo {year} {2003})\ pp.\ \bibinfo {pages} {xiv+341}\BibitemShut
  {NoStop}%
\end{thebibliography}
\end{document}